\documentclass[11pt,letterpaper]{article}

\usepackage{tikz}
\usepackage[margin=1in]{geometry}

\newcommand{\vol}{\mathrm{vol}}

\renewcommand{\phi}{\varphi}

\newcommand{\N}{\mathbb{N}}
\newcommand{\R}{\mathbb{R}}

\newcommand{\cK}{\mathcal{K}}

\newcommand{\cF}{\mathcal{F}}

\def\ds1{\mathds{1}}
\renewcommand{\epsilon}{\varepsilon}

\renewcommand{\tilde}{\widetilde}

\newlength{\minipagewidth}
\newcommand{\beq}{\begin{equation}}
\newcommand{\eeq}{\end{equation}}

\newcommand{\beqa}{\begin{eqnarray}}
\newcommand{\eeqa}{\end{eqnarray}}

\newcommand{\beqan}{\begin{eqnarray*}}
\newcommand{\eeqan}{\end{eqnarray*}}

\def\ba#1\ea{\begin{align*}#1\end{align*}} 
\def\banum#1\eanum{\begin{align}#1\end{align}} 

\usepackage{etoolbox}

\makeatletter

\newcommand{\define}[4][ignore]{%
  \ifstrequal{#1}{ignore}{}{
  \@namedef{thmtitle@#2}{#1}}%
  \@namedef{thm@#2}{#4}%
  \@namedef{thmtypen@#2}{lemma}%
  \newtheorem{thmtype@#2}[theorem]{#3}%
  \newtheorem*{thmtypealt@#2}{#3~\ref{#2}}%
}

\newcommand{\state}[1]{%
  \@namedef{curthm}{#1}
  \@ifundefined{thmtitle@#1}{
  \begin{thmtype@#1}
    }{
  \begin{thmtype@#1}[\@nameuse{thmtitle@#1}]
  }
    \label{#1}
    \@nameuse{thm@#1}
  \end{thmtype@#1}
  \@ifundefined{thmdone@#1}{
  \@namedef{thmdone@#1}{stated}%
  }{}
}

\newcommand{\restate}[1]{%
  \@namedef{curthm}{#1}
  \@ifundefined{thmtitle@#1}{
    \begin{thmtypealt@#1}
    }{
  \begin{thmtypealt@#1}[\@nameuse{thmtitle@#1}]
  }
    \@nameuse{thm@#1}
  \end{thmtypealt@#1}
  \@ifundefined{thmdone@#1}{
  \@namedef{thmdone@#1}{stated}%
  }{}
}

\newcommand{\thmlabel}[1]{
  \@ifundefined{thmdone@\@nameuse{curthm}}{\label{#1}
    }{\tag*{\eqref{#1}}}
}
\makeatother

\usepackage[lined,boxed,ruled,norelsize,algo2e,linesnumbered]{algorithm2e}
\usepackage[hidelinks]{hyperref}
\usepackage{url}
\usepackage[numbers]{natbib}
\usepackage{graphicx,subfigure,amsmath,amssymb,amsfonts,bm,epsfig,epsf,url,dsfont,color}

\usepackage{amsthm}
\newtheorem{theorem}{Theorem}[section]
\newtheorem{conjecture}{Conjecture}[section]
\newtheorem{lemma}[theorem]{Lemma}

\DeclareMathOperator*{\E}{\mathbb{E}}

\title{Competitively Chasing Convex Bodies}

\def\And{\and}
\author{S\'ebastien Bubeck\\
       Microsoft Research\\
			 \And
Yin Tat Lee \thanks{Research supported in part by NSF Awards CCF-1740551, CCF-1749609, and DMS-1839116.}\\
       University of Washington \& Microsoft Research\\
			 \And
       Yuanzhi Li \thanks{This work was done while Y. Li and M. Sellke were at Microsoft Research.}  \\
      Stanford University\\
			 \And
      Mark Sellke \footnotemark[2] \\
      Stanford University \\
}

\begin{document}

\maketitle

\begin{abstract}
Let $\cF$ be a family of sets in some metric space. In the {\em $\cF$-chasing} problem, an online algorithm observes a request sequence of sets in $\cF$ and responds (online) by giving a sequence of points in these sets. The movement cost is the distance between consecutive such points. The competitive ratio is the worst case ratio (over request sequences) between the total movement of the online algorithm and the smallest movement one could have achieved by knowing in advance the request sequence. The family $\cF$ is said to be {\em chaseable} if there exists an online algorithm with finite competitive ratio. In 1991, Linial and Friedman conjectured that the family of convex sets in Euclidean space is chaseable. We prove this conjecture.
\end{abstract}

\section{Introduction}
Let $\cK$ denote the family of convex sets in $\R^d$, and $\cK^*$ the set of strings with alphabet $\cK$. A map $S :  \cK^{*} \rightarrow \R^d$ is called an {\em online} selector if for any $t \in \N$ and $(K_1,\hdots,K_t) \in \cK^t$ one has $S(K_1,\hdots,K_t) \in K_t$. In the convex body chasing problem, introduced in \cite{FL93}, the performance of an online selector on a sequence of convex sets $\mathbf{K} \in \cK^{\N}$ and a starting state $x_0 \in \R^d$ is measured through its {\em movement cost} (with $S(\emptyset) :=x_0$):
\[
\mathrm{cost}_S(\mathbf{K}) := \sum_{t \geq 1} \| S(\mathbf{K}_1,\hdots, \mathbf{K}_{t-1}) - S( \mathbf{K}_1, \hdots, \mathbf{K}_{t}) \| \,,
\]
where $\| \cdot \|$ denotes the Euclidean norm. The cost of the {\em offline optimum}, denoted $\mathrm{cost}^*(\mathbf{K})$, is defined as the infimum of the above quantity over all $S$, or equivalently (with $o_0 = x_0$):
\[
\mathrm{cost}^*(\mathbf{K}) = \inf_{(o_t) \in \mathbf{K}} \sum_{t \geq 1} \|o_{t-1} - o_t\| \,.
\]
We say that an online selector $S$ is $\omega$-competitive if for any $\mathbf{K} \in \cK^{\N}$ one has
\[
\mathrm{cost}_S(\mathbf{K}) \leq \omega \cdot \mathrm{cost}^*(\mathbf{K}) \,.
\]

\begin{conjecture}[\cite{FL93}]
For any $d \in \N$, there exists an $\omega$-competitive online selector for some $\omega > 1$.
\end{conjecture}
In \cite{FL93} this conjecture was proved for $d=2$, and it remained open until now for any $d \geq 3$. In this paper we resolve this conjecture affirmatively and obtain an exponential scaling on the dimension for the competitive ratio (it is known from \cite{FL93} that the competitive ratio has to be $\Omega(\sqrt{d})$):
\begin{theorem} \label{thm:main}
For any $d \in \N$, there exists a $2^{30 d}$-competitive online selector.
\end{theorem}

\subsection{Motivation}
The original motivation for convex body chasing was to better understand the influence of geometric structure in the general chasing problem. Indeed the most central problem in online algorithms, $k$-server (\cite{MMS90}), can be viewed as chasing sets with a certain combinatorial structure. However, while there has been lots of progress on the $k$-server problem since the late eighties (see \cite{BCLLM18, Lee18} for the current state of the art), the convex body chasing problem has been essentially unstirred (only some very special cases have been solved, see related works section below). Fortunately, more convincing applications for convex body chasing have been proposed in the last decade, in particular by considering the generalized problem of {\em chasing convex functions}.

In convex function chasing, a request corresponds to a convex cost function $f_t : \R^d \rightarrow \R_+ \cup \{+\infty\}$, instead of merely a convex set as in convex body chasing. In addition to the movement cost between two consecutive points $x_{t-1}$ and $x_t$, there is now also a service cost of $f_t(x_t)$. It is easy to see that chasing convex functions in dimension $d$ can be reduced to chasing convex bodies in dimension $d+1$ (one can simply replace a function request $f_t$ by requesting the epigraph of $f_t$ followed by requesting $\R^d \times \{0\}$). In fact \cite{Anto16} even give a (more complicated) reduction that does not require going up in dimension. This more general setting can now model various problems in resource management (e.g., powering data centers \cite{lin13}). Interestingly it also applies naturally to online machine learning. In the latter case the movement cost is not explicitly part of the problem, however it arises naturally from the uncertainty in the online learning model, namely the fact that the decision point $x_t$ at time $t$ is evaluated on the cost function $f_{t+1}$ at time $t+1$. In particular for Lipschitz cost functions the discrepancy between $f_{t+1}(x_t)$ and $f_{t+1}(x_{t+1})$ is bounded from above by the movement cost $\|x_t - x_{t+1}\|$. We note that, on the contrary to the classical regret analysis of online convex optimization \cite{Haz16}, here the resulting online machine learning algorithm would be able to track a slowly shifting concept. The price to pay is a multiplicative guarantee instead of an additive guarantee.

\subsection{Related works}
The lack of progress on the Friedman-Linial conjecture prompted the community to consider special cases of convex body chasing, where the request sequence is constrained in some ways. This type of results in fact go back to the original paper, where the question of chasing {\em lines} was resolved. See also \cite{Anto16} and references therein for more results in that vein. More recently the {\em nested} version has gotten some attention (\cite{Bansal18, ABCGL19}): this is the scenario where the request sequence is nested. In a companion paper, \cite{BLLS18}, we solve this simpler problem and obtain a competitive ratio of $O(\sqrt{d \log(d)})$, thus almost matching the $\Omega(\sqrt{d})$ lower bound. Another natural question recently considered is whether dimension-free competitive ratio can be obtained for special type of convex sets. This question was addressed for convex function chasing in \cite{CGW18}, where the dimension-free property was proved for ``linearly growing" functions. See also \cite{BCN14} for logarithmic competitive ratio with linear costs and covering LP type sets.

\subsection{Notation and convex geometry reminders}
We denote $B(x,R)$ for a Euclidean ball centered at $x$ and of radius $R$, $\mathring{B}(x,R)$ for the corresponding open ball, and $B:=B(0,1)$. The minimum width of a convex set $K$ is denoted $\delta(K)$, and the centroid is denoted $\mathrm{cg}(K)$. We also denote $P_V$ for the projection on a subspace $V$, and $\mathrm{dist}(x,\Omega)$ for the Euclidean distance between a point $x \in \R^d$ and a set $\Omega \subset \R^d$. In order to emphasize key aspects we will sometimes use the $O_d$ and $\Omega_d$ notation to hide dimension-dependent constants (eventually all constants are made explicit).

\begin{lemma} \label{lem:grun}
Let $K$ and $L$ be two convex bodies such that $\mathrm{cg}(K) \not\in L$. Then for any $\epsilon >0$ one has
\[
\mathrm{vol}(K \cap (L + \epsilon B)) \leq \left(1 - \frac{1}{e} + \frac{2 d (d+1)}{\delta(K)} \epsilon \right) \mathrm{vol}(K) \,.
\]
\end{lemma}

\begin{proof}
The approximate Gr\"unbaum's inequality states that if $K$ is isotropic then the volume decreases by $1 - 1/e + \epsilon$, \cite[Theorem~3]{BV04}. For non-isotropic $K$ this implies that the volume decreases by $1-1/e + \frac{\epsilon}{\sqrt{\lambda_{\min}}}$ where $\lambda_{\min}$ is the smallest eigenvalue of the covariance matrix of $K$. It only remains to observe that
\begin{equation} \label{eq:triv}
\delta(K) \leq 2 \sqrt{d (d+1) \lambda_{\min}} \,, 
\end{equation}
which is a consequence of the fact that an isotropic convex body is included in a ball of radius $\sqrt{d (d+1)}$ (\cite[Theorem~4.1]{KLS95}).
\end{proof}

\begin{lemma} \label{lem:smallball}
A convex body $K$ contains a ball of radius $\frac{\delta(K)}{2 d (d+1)}$.
\end{lemma}

\begin{proof}
Use that an isotropic body contains a ball (\cite[Theorem~4.1]{KLS95}) together with \eqref{eq:triv}.
\end{proof}

\newcommand{\rn}{r^{\mathrm{new}}}
\newcommand{\zn}{z^{\mathrm{new}}}
\newcommand{\on}{o^{\mathrm{new}}}
\newcommand{\Phin}{\Phi^{\mathrm{new}}}

\section{Proof skeleton: a multiscale chase} \label{sec:skeleton}
Our proposed online selector works in phases, where a phase corresponds to a block of requests. For each phase $b \in \N$, a center $z_b \in \R^d$ and a scale $r_b \geq 0$ will be chosen (with $r_1 = \mathrm{dist}(x_0, K_1)$ and $z_1 = x_0$). As the phase $b$ will always be clear from context, we drop the subscript and simply write $r \equiv r_b$ and $z \equiv z_b$. To describe the update procedure of those parameters we will use the notation $\rn \equiv r_{b+1}$ and $\zn \equiv z_{b+1}$. We adopt similar notation for other parameters. The $t^{th}$ request in a phase is denoted $K_t$, and the online selector's response is $x_t$. It will be useful for us to respond with a point $x_t$ possibly outside of the request $K_t$, in which case the corresponding movement cost is at most $\|x_{t-1} - x_t\| + 2 \mathrm{dist}(x_t, K_t)$. Without loss of generality we can assume $x_{t-1} \not\in K_t$.

\subsection{Target properties} \label{sec:phases}
The general idea is to ensure the following two properties in a phase:
\begin{enumerate}
\item[(i)] The online selector's total movement is $O_d(r)$. As a first step to ensure this, the online selector will remain in the ball $B(z, R)$ during the duration of the phase (for some $R= O_d(r)$ to be defined later).
\item[(ii)] If the optimal selector's movement is $\leq r$, then the ``distance" between the online selector and the optimal selector gets reduced by $\Omega_d(r)$ during the phase.
\end{enumerate}
With these two properties an estimate on the competitive ratio follows from a standard potential argument. Indeed, in each phase either the optimal selector pays $\Omega(\text{online cost})$, or some potential gets decreased by $\Omega(\text{online cost})$. See Lemma \ref{lem:conc} for the details.
\newline

In Subsection \ref{sec:padded} we explain the stopping condition for a phase, as well as how to update the parameters $r$ and $z$. In Subsection \ref{sec:potential} we give the notion of ``distance" (i.e., the potential) that we use to satisfy the Property (ii). In Subsection \ref{sec:cg} we explain how to respond to requests during a phase to satisfy Property (i). There, a critical and non-trivial ``pancake issue" arises, and we devote Section \ref{sec:orth} to resolving this issue.

\subsection{Tracking a small cost optimal selector} \label{sec:padded}
The online selector keeps track of a convex set $\Omega_t \subset B(z,R)$ that contains all possible locations visited (during the current phase) by optimal selectors that:
\begin{enumerate}
\item[(i)] stay in $B(z,R)$ during the phase, and 
\item[(ii)] have movement cost $\leq r$ during the phase. 
\end{enumerate}
An example for such a set would be to take the intersection of {\em padded requests}:
\begin{equation} \label{eq:padded}
\tilde{\Omega}_t = B(z,R) \cap \bigcap_{s \leq t} (K_s + r B) \,.
\end{equation}
Indeed we know that a selector satisfying (i) and (ii) must necessarily lie in $\tilde{\Omega}_t$ (and furthermore $\tilde{\Omega}_t$ is convex). Due to the ``pancake issue'' (see Subsection \ref{sec:cg}) it will turn out that \eqref{eq:padded} is too coarse for our needs. However for now the reader is encouraged to think $\Omega_t \simeq \tilde{\Omega}_t$. We will also require the following condition for $\Omega_t$ which is easy to ensure:
\begin{enumerate}
\item[(iii)] if there is no path in $\Omega_t$ satisfying (i) and (ii) then $\Omega_t = \emptyset$.
\end{enumerate}

The general idea is now to ``follow" $\Omega_t$ (see Subsection \ref{sec:cg}) until it becomes small enough. Precisely we stop the phase at the first time $T$ when $\Omega_T$ is included in a ball of radius $\alpha \cdot r$ for some $\alpha \geq 1$ to be defined later. The radius of the localization ball $B(z,R)$ will turn out to be 
\begin{equation} \label{eq:Rdef}
R = 7 \alpha r \,.
\end{equation} 
When the phase stops there are two possibilities: either $\Omega_T \subset \mathring{B}(z,R-r)$ (``localization is not acting'') or $\Omega_T \cap (B(z, R) \setminus \mathring{B}(z,R-r)) \neq \emptyset$ (``localization acting''). For reasons to be explained in Subsection \ref{sec:empty} we also distinguish the case where $\Omega_T=\emptyset$.
\begin{enumerate}
\item If the localization is not acting and $\Omega_T \neq \emptyset$, then we have trapped the optimal selector (provided that it moved $\leq r$, see Lemma \ref{lem:trap} below for the details) and thus we can start a new phase at a lower scale, say $\rn = r/2$. We also choose $\zn$ to be the center of the enclosing ball for $\Omega_T$.
\item If the localization is acting then the optimal selector might be outside of the localized ball $B(z,R)$, and thus we start a new phase at a larger scale, say $\rn = 2 r$
and $\zn =z$. 
We also perform this update when $\Omega_T=\emptyset$.
\end{enumerate}

\begin{figure}
\centering
\begin{tikzpicture}[y=0.80pt, x=0.80pt, yscale=-3, xscale=3]
	\path[draw=black, line width=1pt]
		(36,42) ellipse (25 and 25);
	\path[draw=black, dash pattern=on 5pt off 5pt,line width=1pt]
		(36,42) ellipse (20 and 20);
	\path[draw=black,line width=1pt] (36,17) -- (36,42) 
		node[left, pos=0.6] {$R$} node[right, pos=1] {$z$};
	\path[draw=black,line width=1pt] (56,42) -- (61,42)
		node[above, pos=0.5] {$r$};
	\path[draw=black,line width=1pt]
		(26,40) -- (27,45) -- (23,46) -- (20,43) -- (22,39) -- cycle;
	\path (28,48) node (text5069) {$\Omega_T$};
\end{tikzpicture}
\hspace{3cm}
\begin{tikzpicture}[y=0.80pt, x=0.80pt, yscale=-3, xscale=3]
	\path[draw=black, line width=1pt] 
		(36,42) ellipse (25 and 25);
	\path[draw=black, dash pattern=on 5pt off 5pt,line width=1pt]
		(36,42) ellipse (20 and 20);
	\path[draw=black,line width=1pt] (36,17) -- (36,42) 
		node[left, pos=0.6] {$R$} node[right, pos=1] {$z$};
	\path[draw=black,line width=1pt] (56,42) -- (61,42)
		node[above, pos=0.5] {$r$};
	\path[draw=black,line width=1pt]
		(46,19) -- (44,28) -- (52,32) -- (55,26);
	\path (48,33) node (text5069) {$\Omega_T$};
\end{tikzpicture}
\caption{The phase ends when the convex set $\Omega_t$ is trapped in a ball of radius $\alpha \cdot r$. The localization is not acting on the left figure (condition 1) and is acting on the right figure (condition 2).}
\end{figure}
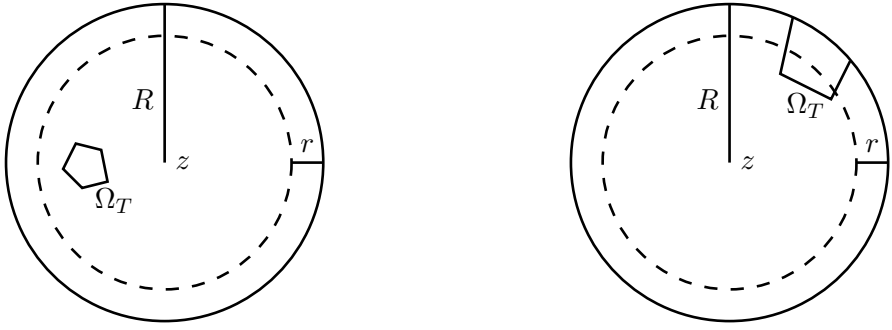

\begin{lemma} \label{lem:trap}
If the phase ends with condition 1 above, then any selector with total movement in the phase $\leq r$ must end the phase in $\Omega_T$.
\end{lemma}

\begin{proof}
Assume on the contrary that there exists a path $(o_t)_{t \in [T]}$ with movement $\leq r$ and that visits $\R^d \setminus B(z,R)$ at some point during the phase (by definition of $\Omega_T$ we know that a path that stays in $B(z,R)$ and with movement $\leq r$ must be in $\Omega_T$). Since $\Omega_T \neq \emptyset$ we also know by property (iii) above that there exists a path $(o_t')_{t \in [T]}$ with movement $\leq r$ that stays in $B(z,R)$, and ends the phase in $B(z,R-r)$ (since $\Omega_T \subset B(z,R-r)$). Let $p \in [0,1]$ be such that $\|p {o}_T' + (1-p) {o}_T - z\| = R-r$. Then clearly $p {o}' + (1-p) {o}$ has movement $\leq r$, stays in $B(z,R)$, yet ends in $\partial B(z, R-r)$ which contradicts the fact that $\Omega_T \subset \mathring{B}(z,R-r)$.
\end{proof}

\subsection{The potential argument} \label{sec:potential}
Let us denote $o$ (resp. $\on$) for the location of the optimal selector at the start (resp. end) of the phase. We consider the following potential: $\Phi:= \max(\|z-o\| - \alpha r, \alpha r)$. Denote $\Phin:= \max(\|\zn-\on\| - \alpha \rn, \alpha \rn)$. 
\begin{lemma} \label{lem:pot}
Let $R = 7 \alpha \cdot r$. 
In each phase 
(as described in Subsection \ref{sec:padded}) one has
\begin{equation} \label{eq:pot1}
\left( 1 + 8.5 \alpha \right) \cdot \mathrm{cost}^*|_{\mathrm{phase}} + \Phi - \Phin \geq \frac{\alpha}{2} r \,, 
\end{equation}
where $\mathrm{cost}^*|_{\mathrm{phase}}$ denotes the cost of the optimal selector during the current phase. 
\end{lemma}

\begin{proof}
We distinguish three cases:
\begin{itemize}
\item If $\mathrm{cost}^*|_{\mathrm{phase}} \geq r$ then we use that (recall that $\|\zn - z\| \leq R = 7 \alpha r$):
\begin{eqnarray*}
\Phin \leq \max ( 2 \alpha r , \|\zn - \on\| ) & \leq & \max(2 \alpha r, \|z - o\| + \|\zn - z\| + \|o - \on\|) \\
& \leq & \Phi + (1+ 8 \alpha) \cdot \mathrm{cost}^*|_{\mathrm{phase}} \,,
\end{eqnarray*}
which proves \eqref{eq:pot1}.

We now assume that $\mathrm{cost}^*|_{\mathrm{phase}} \leq r$.
\item If we end the phase with condition 1 from Subsection \ref{sec:padded} then by Lemma \ref{lem:trap} we have that $\|\zn - \on\| \leq \alpha r$. Thus $\Phin = \alpha r/2$ (indeed $\rn = r/2$), while on the other hand $\Phi \geq \alpha r$. This proves \eqref{eq:pot1}.
\item Finally if we end the phase with condition 2 from Subsection \ref{sec:padded} then we must have $\|z - \on\| \geq R-r - 2 \alpha r \geq 4 \alpha r$, which implies
\[
 \max(\|z-\on\| - \alpha r, \alpha r) = \alpha r+ \max(\|z-\on\| - 2 \alpha r, 2 \alpha r) \geq \alpha r + \Phin \,.
\]
where the last inequality 
uses that $\zn = z$.
Thus we have:
\[
\Phin \leq \max(\|z-\on\| - \alpha r, \alpha r) - \alpha r \leq \|o-\on\| + \Phi - \alpha r \leq \mathrm{cost}^*|_{\mathrm{phase}}  + \Phi - \alpha r \,,
\]
which concludes the proof of \eqref{eq:pot1}.
\end{itemize}
\end{proof}

We note that Lemma \ref{lem:pot} together with Property (i) in Subsection \ref{sec:phases} imply a finite competitive ratio:
\begin{lemma} \label{lem:conc}
Consider a phase-based online selector as described in Subsection \ref{sec:phases} and \ref{sec:padded}. Assume that it has movement at most $\beta r$ per phase (Property (i) in Subsection \ref{sec:phases}), for some $\beta \geq 1$. Then it is $21 \beta$-competitive. 
\end{lemma}

\begin{proof}
Equation \eqref{eq:pot1} implies that (since $\alpha \geq 1$)
\[
r \leq 19 \cdot \mathrm{cost}^*|_{\mathrm{phase}} + \frac{2}{\alpha} (\Phi - \Phin) \,.
\] 
Together with the assumption of $\beta r$-movement per phase, as well as the fact that the potential is non-negative and its initial value is smaller than $\alpha \cdot \mathrm{cost}^*$, one obtains:
\[
\mathrm{cost}_S \leq 21 \cdot \beta \cdot \mathrm{cost}^*  \,.
\]
\end{proof}

\subsection{In-phase movement with center of gravity} \label{sec:cg}
Upon receiving a new request $K_t$ (such that $x_{t-1} \not\in K_t$) we shall update $\Omega_{t-1}$ to $\Omega_t$ which will be a subset of $\Omega_{t-1} \cap (K_t + r B)$. Recall also that a phase ends when $\Omega_t$ is included in a ball of radius $\alpha \cdot r$. Thus the hope is that $x_{t-1}$ is ``deep" inside $\Omega_{t-1}$ so that $\Omega_t$ gets significantly smaller than $\Omega_{t-1}$. Importantly we note that, since we ``cut" with $K_t + r B$ but we only have the guarantee that $x_{t-1} \not\in K_t$, we will also want that $\Omega_{t-1}$ is ``not too small" (the ``pancake issue"), so that cutting $x_{t-1}$ or cutting $r$-close to $x_{t-1}$ has a similar effect.
\newline

Let us consider what happens with the simple proposal of moving to the centroid, that is $x_t= \mathrm{cg}(\Omega_t)$ (provided that $\Omega_t$ is non-empty, see Subsection \ref{sec:empty} for the empty case). Note that $x_t$ is not necessarily in $K_t$, and thus in addition to moving from $x_{t-1}$ to $x_t$ one might incur an extra $r$ movement to actually satisfy the request. Thus after $T$ requests the online selector has moved at most $T \cdot (R+r)$. In fact thanks to Lemma \ref{lem:emptycase} below this bound will hold true even if the phase ends with $\Omega_T = \emptyset$. Now the approximate Gr\"unbaum's inequality (Lemma \ref{lem:grun}) implies that
\begin{equation} \label{eq:minwidthtoensure}
\delta(\Omega_{t-1}) \geq 8 d (d+1) r \Rightarrow \mathrm{vol}(\Omega_t) \leq \frac{9}{10} \mathrm{vol}(\Omega_{t-1}) \,.
\end{equation}

Furthermore, if we can guarantee that the minwidth remains $8 d (d+1)r $, then we have that $\Omega_t$ always contain a ball of radius $4 r$ (see Lemma \ref{lem:smallball}), which in turn means that the total volume decrease is at most $(R / (4r))^d \leq (2 \alpha)^d$ (by \eqref{eq:Rdef}), at which point the phase will stop (we will take $\alpha \geq 4$). Thus we see that the length of a phase is $T \leq d \log_{10/9}(2 \alpha)$, which in turn guarantees that the total movement in a phase is $(R+r) \cdot d \log_{10/9}(2 \alpha) = O_d(r)$.
\newline

We see that the {\em only} remaining difficulty is to ensure that the minwidth of the set $\Omega_t$ remains $\Omega_d(r)$. In our former nested convex body chasing work, \cite{ABCGL19}, we dealt with the small width directions via a recursive argument. Here the situation is much more delicate, and we dedicate the next section to this ``pancake issue". We also note that, at this point of the argument, one has a $\mathrm{poly}(d)$-competitive ratio. The exponential dependency in Theorem \ref{thm:main} comes from the recursive argument (which could potentially be improved).

\subsection{The empty scenario} \label{sec:empty}
When $\Omega_t=\emptyset$ we know that a new phase will begin after the current request $K_t$. However it might be that to satisfy the request $K_t$ one has to pay a movement much larger than $r$. The next lemma shows a simple reduction which allows us to assume that this never happens.

\begin{lemma} \label{lem:emptycase}
Any request sequence can be modified online so that:
\begin{itemize}
\item Each original request $K$ is replaced by a finite nested sequence ending with $K$ (in particular the value of the offline optimum is the same on the original and modified sequences).
\item On the modified sequence, in each phase one has that all requests satisfy $\mathrm{dist}(z, K) \leq R+r$.
\end{itemize}
\end{lemma}

\begin{proof}
Let us consider a request $K$ which violates the second point. Then we know that the current phase will end with condition 2 from Subsection \ref{sec:padded} since $(K+rB) \cap B(z,R) = \emptyset$. In the modified sequence we replace $K$ by $K+ h B$ where $h > 0$ is such that $\mathrm{dist}(z, K+hB)=R+r$, so that the current phase still ends with condition 2 but now the ``virtual'' request $K+hB$ satisfy the second point in Lemma \ref{lem:emptycase}.
 
In the next phase we start by giving tentatively $K$ and repeat the procedure above, that is we potentially replace $K$ by $K+h^{\mathrm{new}} B$ if $\mathrm{dist}(\zn, K) > R + \rn$. The key point is that eventually (in a finite number of steps), the request $K$ will be valid, since $\rn = 2r$ will keep doubling while $\zn=z$ remains constant.
\end{proof}
\section{Induction on dimension argument} \label{sec:orth}
In this section we discuss how to deal with the small width directions in $\Omega_t$. By induction we assume that we have access to an $\omega_k$-competitive algorithm for convex body chasing in $\R^k$, $k < d$ (note that the case $k=1$ is trivial).

Let us denote $V_t$ for a subspace spanned by a maximal set of orthogonal ``small" width directions in $\Omega_t$. More precisely, given $V_t$ and $\Omega_{t+1}$, we define $V_{t+1}$ by the following iterative procedure:
\begin{itemize}
\item $V\leftarrow V_{t}$.
\item While there exists $v\in V^{\perp}$ such that the width of $\Omega_{t+1}$
in the direction $v$ is $\leq \frac{\alpha}{2\dim V^{\perp}}$
\begin{itemize}
\item Set $V\leftarrow\mathrm{span}(V,v)$
\end{itemize}
\item $V_{t+1}\leftarrow V$
\end{itemize}
Note that if $\mathrm{dim}(V_t)=d$ then we know that $\Omega_t$ is included in a hyperrectangle with lengths $\frac{\alpha}{2 d} r, \frac{\alpha}{2 (d-1)} r, \hdots$ and thus it is also included in a Euclidean ball of radius $$\frac{\alpha r}{2}\sqrt{\frac{1}{d^2}+\frac{1}{(d-1)^2}+\cdots+1} \leq \alpha \cdot r.$$ In particular we obtain that $\mathrm{dim}(V_t) = d$ ensures that the global phase ends at time $t$. Thus let us assume $\mathrm{dim}(V_t) < d$. 

\begin{figure}
\centering
\begin{tikzpicture}[y=0.80pt, x=0.80pt, yscale=-3, xscale=3]
	\path[draw=black,line width=1pt] (28,43) ellipse (25 and 25);
	\path[draw=black,dash pattern=on 5pt off 5pt,line width=1pt]
		(26,23) -- (42,61) node[left, pos=0.1] {$V_t^\perp$};
	\path[draw=black,line width=1pt]
		(29,29) -- (26,37) -- (33,51) -- (39,58) -- (41,48) -- (29,29);
	\path (30,55) node {$\Omega_t$};
	\path (50,65) node {$B(z,R)$};
	\path[draw=black,dash pattern=on 5pt off 5pt,line width=1pt]
		(20,48) -- (47,37) node[above, pos=0] {$V_t$};
\end{tikzpicture}
\caption{We use $V_t$ to denote the subspace spanned by small width directions in $\Omega_t$}
\end{figure}
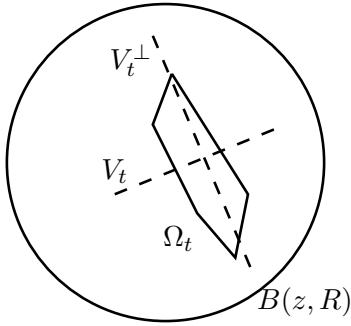

\subsection{The nested intuition} \label{sec:nested}
To set the stage, let us briefly recall the {\em nested} convex body chasing strategy we proposed in our prior work \cite{ABCGL19}. In this subsection we view the sequence $\Omega_t, \Omega_{t+1}, \hdots$ as a nested convex body chasing request sequence.
The entire algorithm and analysis can be summarized in the following two points:
\begin{enumerate}
\item Let $x_t$ be any point such $P_{V_t^{\perp}} x_t = \mathrm{cg}(P_{V_t^{\perp}} \Omega_t)$. One can now play the $\mathrm{dim}(V_t)$-algorithm on $x_t + V_t$, until one gets an empty set on that affine subspace, that is a time $T$ such that $\Omega_T \cap (x_t + V_t) = \emptyset$. For notation simplicity let us say that $T=t+1$. Crucially at this point one obtains $P_{V_t^{\perp}} x_t \not\in P_{V_t^{\perp}} \Omega_{t+1}$, so that by Gr\"unbaum's inequality:
\begin{equation} \label{eq:cutinnested}
\mathrm{vol}(P_{V_t^{\perp}} \Omega_{t+1}) \leq (1-1/e) \cdot \mathrm{vol}(P_{V_t^{\perp}} \Omega_{t}) \,.
\end{equation}
\item The natural potential to track progress is $\Psi_t := \mathrm{vol}(P_{V_t^{\perp}} \Omega_t)$. We note that while $k=\mathrm{dim}(V_t^{\perp})$ stays fixed, the potential drops by $(1-1/e)$ at each time step, starting from a value of at most $R^k = (7 \alpha r)^k$ and ending with a value of at least $(r \alpha / (4 k^2 (k+1)))^k$ (by Lemma \ref{lem:smallball} and the minimum width $\geq \alpha/(2k)$). In particular the number of steps where the dimension of $V_t^{\perp}$ remains $k$ is at most $k \cdot \log_{e/(e-1)}(28 k^2 (k+1))$.
\end{enumerate}
Combining both points we would obtain a $O(d^2 \log d)$-competitive algorithm for nested convex body chasing. We note that a slightly more careful analysis of $\Psi_t$ easily yields a $O(d \log d)$-competitive analysis, but it is irrelevant for the purpose of the present paper.

The essence of the difficulty in the non-nested case (in fact the {\em only} remaining obstacle at this point) is that one cannot afford to wait until a time where $\Omega_{T} \cap (x_t + V_t)$ gets empty. Indeed the lower-dimensional algorithm in the affine subspace is only competitive with respect to an optimal selector that stays within that affine subspace. What if such an optimal selector has in fact much worse performance than an unconstrained selector? This is indeed a critical situation for convex body chasing: all requests could have an intersection point far away from the current affine subspace, so that the lower-dimensional algorithm pays a lot while the optimal selector pays nothing. Perhaps the most important conceptual contribution of this work is to realize that in this case one can extract from the request sequence a cutting plane that is {\em parallel} to the current affine subspace, and thus actually obtain a large volume decrease in the orthogonal subspace just as in \eqref{eq:cutinnested}. This is detailed in Subsection \ref{sec:findorth}.

\subsection{Leveraging lower dimensional selectors} \label{sec:findorth}
Let us assume that we start a lower-dimensional phase at some time $t_0$ in the current global phase, by which we mean that $\mathrm{dim}(V_{t_0}) \in \{1,\hdots, d-1\}$. Denote $k=\dim(V_{t_0}^{\perp})$.
Just like in the nested case we first move to any point $x_{t_0}$ such that $P_{V_{t_0}^{\perp}} x_{t_0} = \mathrm{cg}(P_{V_{t_0}^{\perp}} \Omega_{t_0})$, and we now run the $(d-k)$-dimensional online selector in $x_{t_0} + V_{t_0}$. We now stop the lower-dimensional phase when the movement cost of the $(d-k)$-dimensional online selector reaches $\zeta_k \cdot \omega_{d-k} \cdot r$, for some $\zeta_k$ to be defined later. 
\newline

For notation simplicity we will assume that $x_{t_0}=0$, and write $V \equiv V_{t_0}$, $\Omega \equiv \Omega_{t_0}$.
Again with an abuse of notation we will now only use a ``local time" for the current lower-dimensional phase. In particular $K_1, K_2, \hdots$ now denotes the sequence of requests in this lower-dimensional phase. 
\newline

At the beginning of the lower-dimensional phase we assume by induction that $\Omega$ satisfies the following case analysis (recall property (i) and (ii) in Subsection \ref{sec:padded}). In the global phase, any selector must either:
\begin{itemize}
\item have a total cost $\geq r$;
\item have visited a point outside of $B(z,R)$ ; or finally
\item remain in $\Omega$. 
\end{itemize}
Our hope is that by the end of the phase (i.e., when the lower-dimensional algorithm has moved at least $\zeta_k \cdot r$) we will be able to refine this set $\Omega$ of possible locations for a localized small cost selector into a set $\Omega^{\mathrm{new}}$ such that
\begin{equation} \label{eq:targetcut}
\mathrm{vol}(P_{V^{\perp}} \Omega^{\mathrm{new}}) \leq \frac{9}{10} \cdot \mathrm{vol}(P_{V^{\perp}} \Omega) \,.
\end{equation}

\begin{lemma}
Assume that \eqref{eq:targetcut} holds true for any $k \in [d]$. Then, in a global phase, the total number of calls to the $(d-k)$-dimensional online selector\footnote{A call to a $0$-dimensional selector simply means receiving a request.} is at most $50 k \log(k+1)$.
\end{lemma}

\begin{proof}
The argument is the same as in the second point of Subsection \ref{sec:nested}. Namely we consider the potential $\Psi := \mathrm{vol}(P_{V^{\perp}} \Omega)$. While $k=\mathrm{dim}(V^{\perp})$ stays fixed, the potential drops by $9/10$ with each call to the $(d-k)$-dimensional online selector. Furthermore the potential starts from a value of at most $R^k = (7 \alpha r)^k$, and ends with a value of at least $(r \alpha / (4 k^2 (k+1)))^k$ (by Lemma \ref{lem:smallball}  and the minimum width $\geq \alpha/(2k))$). Thus the number of calls is at most $k \cdot \log_{10/9}(28 k^2 (k+1)) \leq 50 k \log(k+1)$.
\end{proof}

This in turns give the following (recall that after each call to a lower-dimensional selector one has to move to the center of gravity of $P_{(V^{\mathrm{new}})^{\perp}} \Omega^{\mathrm{new}}$, which adds a movement of $R+r$):

\begin{lemma} \label{lem:betadef}
Assume that \eqref{eq:targetcut} holds true for any $k \in [d]$. Then the total movement in a global phase is at most $\beta \cdot r$, with
\[
\beta = \sum_{k=1}^{d} 50 k \log(k+1) \cdot (\zeta_k \cdot \omega_{d-k} + 7 \alpha + 1) \,.
\]
\end{lemma}

Thus we see that it only remains to estimate the value of $\zeta_k$ (and of $\alpha$) so that one can guarantee \eqref{eq:targetcut}. Note for example that \eqref{eq:targetcut} is satisfied for $k=d$ (i.e., $V=\{0\}$, $\delta(\Omega) \geq \frac{\alpha}{2 d}$, and there is no call to a lower-dimensional online selector) with $\Omega^{\mathrm{new}} = \Omega \cap \{K + r B\}$ provided that (recall \eqref{eq:minwidthtoensure})
\begin{equation} \label{eq:alphacond}
\frac{\alpha}{2 d} \geq 8 d(d+1) \,.
\end{equation}

\subsection{A refined localization of small cost optimal selectors}
Consider the set
\[
\Omega'_t = \left\{y \in \Omega \text{ s.t. } \exists (y_s)_{s \in \{1,\hdots,t\}} \text{ with } y_1 \in \Omega, y_{t} = y , y_s \in K_s \cap \Omega, \text{ and } \sum_{s=1}^{t-1} \|y_s - y_{s+1}\| \leq r \right\} \,.
\]
Crucially we note that $\Omega'_t$ contains the set of possible locations at time $t$ for a selector with movement $\leq r$ and that starts the lower-dimensional phase in $\Omega$. Furthermore $\Omega'_t$ is convex (by convexity of $\Omega$, $K_s$ and the function $y \mapsto \|y\|$).
Our tentative choice is $\Omega^{\mathrm{new}} := \Omega \cap (\Omega_{T}' + r B)$ at the time $T$ when the lower-dimensional phase ends. Note that $\Omega^{\mathrm{new}}$ is indeed a valid refinement of $\Omega$ (in the sense that any selector with movement $\leq r$ and that remains in $B(z,R)$ must in fact remain in $\Omega^{\mathrm{new}}$ for the whole global phase). 
We now want to show that \eqref{eq:targetcut} holds true.
\newline

The key non-trivial observation is the following: during the lower-dimensional phase, if the best selector that remains in $V$ has total movement $\eta r$ (for some $\eta = \Theta_k(1)$), then it must be that there exists $x \in V^{\perp}$ with $\|x \| \leq \gamma r$ (for some $\gamma = \Theta_d(1)$) and $x \not\in \Omega_t'$ (the precise formulation is given in Lemma \ref{lem:key} below). This means that for $\zeta_k \geq \eta$ we know that $x \not\in \Omega^{\mathrm{new}}$, and thus since $x$ is close (for $\gamma$ small enough compared to $\delta(P_{V^{\perp}} \Omega) \geq \frac{\alpha}{2 k}$) to the center of gravity of $P_{V^{\perp}} \Omega$ (which is assumed to be $0$), we expect that \eqref{eq:targetcut} will be satisfied, see Subsection \ref{sec:final} for the details. 

In words this key lemma says that: if when restricted to a subspace $V$ one has to move a lot to satisfy the requests, while at the same time there was a way to satisfy all those requests without moving much, then it must be that we can discard a large portion of $V^{\perp}$ as possible good locations.

\begin{figure}
\centering
\begin{tikzpicture}[y=0.80pt, x=0.80pt, yscale=-0.28, xscale=0.28]
  \path[draw=black,nonzero rule,line width=1pt]
    (357.5805,775.2712) ellipse (7.2813cm and 2.7165cm);
  \path[fill=black,line width=0.800pt]
    (580.5809,880) node[above right] {$\mathbb{S}\subset V^\perp$};

  \begin{scope}[cm={{0.99911,-0.04778,0.05321,1.42992,(-43.46275,-337.56072)}}]
    \path[draw=black,line width=1pt] (105.8396,685.4081) .. controls (131.3621,668.0766)
      and (88.7677,651.7594) .. (80.4664,640.6089) .. controls (80.4664,640.6089)
      and (35.4565,592.4751) .. (111.0753,596.7950) node[pos=1, above right] {$y_t(\theta'')$};
    \path[draw=black,line width=1pt] (107.9899,795.7922) .. controls (119.0945,752.0128)
      and (76.7040,757.2858) .. (102.8064,733.5553) node[pos=0, below left] {$\theta''$};
    \path[draw=black,line width=1pt] (106.2621,685.1965) .. controls (69.4004,710.1961)
      and (124.4082,713.1615) .. (101.9425,734.4398);
  \end{scope}
    \begin{scope}[cm={{1.39183,0.19144,-0.18412,1.44715,(-96.83701,-186.04231)}}]
      \path[draw=black,line width=1pt] (521.8124,650.9203) .. controls (513.4745,645.9639)
        and (536.3585,639.4595) .. (522.3293,628.4293) node[pos=0, below] {$\theta'$};
      \path[draw=black,line width=1pt] (522.6713,629.0417) .. controls (509.1880,616.5663)
        and (490.7714,580.7584) .. (522.6713,581.5261);
      \path[draw=black,line width=1pt] (522.4345,581.5261) .. controls (552.6280,589.5182)
        and (503.8719,490.0412) .. (530.8429,508.9569);
      \path[draw=black,line width=1pt] (530.5377,508.7079) .. controls (536.1775,516.3611)
        and (535.0537,505.2036) .. (531.1515,493.7712) node[pos=0, above right] {$y_t(\theta')$};
    \end{scope}
  \begin{scope}[cm={{0.99581,0.0915,-0.0915,0.99581,(95.9807,72.71819)}}]
    \path[draw=black,line width=1pt] (395.1538,501.7977) .. controls (426.1896,479.3406)
      and (369.7113,460.0271) .. (389.5979,439.5750) node[pos=1, above right] {$y_t(\theta)$};
    \path[draw=black,line width=1pt] (389.2602,440.0534) .. controls (387.9991,437.5312)
      and (421.2082,424.2047) .. (388.7047,404.6382);
    \path[draw=black,line width=1pt] (395.1582,577.0540) .. controls (445.0347,551.1083)
      and (371.6647,519.2254) .. (396.0902,501.5638) node[pos=0, below] {$\theta$};
  \end{scope}
  \path[draw=black,line width=1pt] (436.6312,683.2994) ellipse (0.1447cm and 0.1578cm);
    \path[draw=black,dash pattern=on 5pt off 5pt,line width=1pt]
      (108.0361,511.4655) -- (354.1800,566.1176);
    \path[draw=black,dash pattern=on 5pt off 5pt,line width=1pt]
      (359.4808,565.6056) -- (551.0490,630.3000);
  \path[draw=black,dash pattern=on 5pt off 5pt,line width=1pt]
    (360,775.5600) -- (360,380) node[pos=1,left] {$V$} node[pos=0.45,left] {$Y_t$};
  \path[draw=black,dash pattern=on 9.60pt off 9.60pt,line join=miter,line
    cap=butt,miter limit=4.00,draw opacity=0.784,even odd rule,line width=1pt]
    (359.7431,562.1423) .. controls (444.5530,510.8836) and (444.5530,510.8836) ..
    (444.5530,510.8836);
\end{tikzpicture}
\hspace{2cm}
\begin{tikzpicture}[y=0.80pt, x=0.80pt, yscale=-0.700000, xscale=0.700000]
  \begin{scope}
    \clip  (386.6490,188.9006) rectangle  (694.2405,362.6468);
      \path[cm={{0.99521,0.09772,-0.09772,0.99521,(0.0,0.0)}},opacity=0.368,fill=black] (152.2230,237.3448) rectangle (732.0505,803.0303);
      \path[cm={{0.97646,-0.2157,0.2157,0.97646,(0.0,0.0)}},opacity=0.368,fill=black] (60.7537,-191.2056) rectangle (640.5813,374.4798);
      \path (650,230) node[above right]  {$K_t$};
      \path (650,330) node[above right] (text5568) {$K_{t+1}$};
      \path[draw=black,line width=1pt]
        (540,275) ellipse (3.8509cm and 0.1730cm);
      \begin{scope}[cm={{0.14214,0.01393,-0.01306,0.15166,(490.49185,197.91506)}}]
        \path[draw=blue,line width=1pt] (395.1538,501.7977) .. controls (426.1896,479.3406)
          and (369.7113,460.0271) .. (389.5979,439.5750);
        \path[draw=blue,line width=1pt] (389.2602,440.0534) .. controls (387.9991,437.5312)
          and (421.2082,424.2047) .. (388.7047,404.6382);
        \path[draw=blue,line width=1pt] (395.1582,577.0540) .. controls (445.0347,551.1083)
          and (371.6647,519.2254) .. (396.0902,501.5638);
      \end{scope}

      \begin{scope}[shift={(0,-5)}, cm={{0.13033,0.00669,-0.01197,0.0728,(389.11665,243.28429)}}]
        \path[draw=red,line width=1pt] (395.1538,501.7977) .. controls (426.1896,479.3406)
          and (369.7113,460.0271) .. (389.5979,439.5750) node[pos=1, above] {$\theta$};
        \path[draw=red,line width=1pt] (389.2602,440.0534) .. controls (387.9991,437.5312)
          and (421.2082,424.2047) .. (388.7047,404.6382);
        \path[draw=red,line width=1pt] (395.1582,577.0540) .. controls (445.0347,551.1083)
          and (371.6647,519.2254) .. (396.0902,501.5638) node[pos=0, below] {$y_t(\theta)$};
      \end{scope}

      \path[draw=black,dash pattern=on 5pt off 5pt,line width=1pt]
        (540,327.3344) -- (540,204.9696)  node[pos=1,left] {$V$};
      \path (655,270) node[above] {$\mathbb{S}\subset V^\perp$};
    \end{scope}

\end{tikzpicture}
\caption{The left figure illustrates that one can construct a selector $Y_t$ in $V$ by taking average of selectors $y_t(\theta)$ starting at $\theta\in V^\perp$ with $\|\theta\| = \gamma r$.
The right figure gives a case where there is a good selector (red) outside of $V$ but all selectors (blue) in $V$ are bad. In this case, we can discard a large portion of $V^\perp$.}
\end{figure}
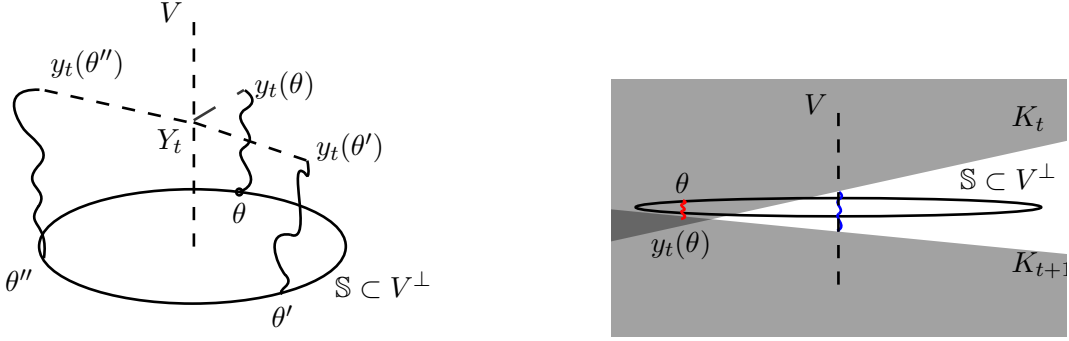

\begin{lemma} \label{lem:key}
Let $V \subset \R^d$ be a $n-k$ dimensional linear subspace. Consider the cylinder $C = \{y\in \R^d : \|P_{V} y\| \leq \gamma' r\}$ and let
\[
\tilde{\Omega} = \left\{y \in C \text{ s.t. } \exists (y_t)_{t \in \{1,\hdots,T\}} \text{ with } y_{T} = y , y_t \in K_t \cap C, \text{ and } \sum_{t=1}^{T-1} \|y_t - y_{t+1}\| \leq r \right\} \,.
\]
Assume that $\mathbb{S} := \{\theta \in V^{\perp} : \|\theta\|=\gamma r\} \subset P_{V^{\perp}} \tilde{\Omega}$, where $\gamma \geq 12 k$. Then there exists $(Y_t)_{t \in \{1,\hdots,T\}}$ with $Y_t \in K_t \cap V$ and $\sum_{t=1}^{T-1} \|Y_t - Y_{t+1}\| \leq \left(1+ \frac{2+4\gamma'}{\gamma} k \right) r$.
\end{lemma}

\begin{proof}
For any $\theta \in \mathbb{S}$, let us consider the path $(y_t(\theta))$ witnessing that $\theta \in \tilde{\Omega}$. Denote $u_t(\theta) = P_{V^{\perp}} y_t(\theta)$ and $v_t(\theta) = P_{V} y_t(\theta)$. 
Our tentative path $(Y_t)$ will be taken of the form:
\[
Y_t = \E [ \mu_t(\theta) y_t(\theta) ] \,,
\]
where the expectation is with respect to a random $\theta$ in $\mathbb{S}$, and $\mu_t : \mathbb{S}  \rightarrow \R_+$ is a non-negative weight function satisfying $\E [\mu_t(\theta)] = 1$ (so that $Y_t \in K_t$) and $\E[ \mu_t(\theta) u_t(\theta)] = 0$ (so that $Y_t \in V$). Let us define $m_t = \E[ u_t(\theta) ]$ and $\Sigma_t = \E[ (u_t(\theta) - m_t) (u_t(\theta) - m_t)^{\top}]$ (which we view as a linear operator on $V^{\perp}$). We propose to take the following weight function (which can be derived as the minimum energy function satisying the constraints):
\[
\mu_t(\theta) = 1 + (m_t - u_t(\theta))^{\top} \Sigma_t^{-1} m_t \,.
\]
Assuming that $\mu_t$ is well-defined (i.e., $\Sigma_t$ is invertible as an operator on $V^{\perp}$) one easily checks that $\E [\mu_t(\theta)] = 1$ and $\E[ \mu_t(\theta) u_t(\theta)] = 0$. Let us now prove that $\mu_t$ is indeed well-defined, and also that it is non-negative. 
\newline

First note that by definition, since $y_T(\theta) = \theta$, one has $m_T=0$ and $\Sigma_T = (\gamma r / \sqrt{k})^2 \mathrm{Id}$. Let $\xi_t(\theta) = u_t(\theta) - m_t - \theta$, and note that $\|\xi_t\| \leq 2 r$ (since the total movement of the $y_t(\theta)$ path is at most $r$). With this notation we have:
\[
\Sigma_t = (\gamma r / \sqrt{k})^2 \mathrm{Id} + \E[\xi_t(\theta) \xi_t(\theta)^T + \theta \xi_t(\theta)^{\top} + \xi_t(\theta) \theta^{\top}] \,.
\]
Note that the spectral norm of the second term is bounded by $4 r^2 + 4 \gamma r^2$ and thus we see that:
\begin{equation} \label{eq:lbsigma}
\Sigma_t \succeq (\gamma^2 / k - 4(\gamma+1)) r^2 \mathrm{Id} \succeq \frac{\gamma^2 r^2}{2k} \mathrm{Id} \,,
\end{equation}
where the second inequality uses that $\gamma \geq 12k$. In particular we see that $\Sigma_t \succ 0$ and thus $\mu_t$ is well-defined.
Next using that 
\begin{equation} \label{eq:ubsimple}
\|m_t\| \leq r \text{ and } \|u_t(\theta)\|\leq (\gamma+1) r \,,
\end{equation} 
we obtain: 
\[
|u_t(\theta)^{\top} \Sigma_t^{-1} m_t| \leq 2 k (\gamma + 1) / \gamma^2 \,.
\]
Thus we have that $\mu_t(\theta) \geq 0$ provided that $\gamma^2 \geq 2k (\gamma+1)$
(which is true since $\gamma \geq 12 k$). We note that \eqref{eq:lbsigma} and \eqref{eq:ubsimple} also give:
\begin{equation} \label{eq:magmu}
\mu_t(\theta) \leq 1+ 2 k (\gamma+2) / \gamma^2 \,.
\end{equation}
\newline

Finally it remains to estimate the movement of $Y_t$. For this it will be slightly more convenient to assume that the path is defined in continuous time, so that we want to estimate $\int_0^T \|\frac{d}{dt} Y_t\| dt$. Note that since $P_{V^{\perp}} Y_t = 0$, in fact we only need to estimate the movement of $P_V Y_t = \E[ \mu_t(\theta) v_t(\theta) ]$. Recall that by assumption $v_t(\theta)$ has movement $\leq r$, and thus with \eqref{eq:magmu} one obtains:
\begin{equation} \label{eq:int1}
\int_0^T \| \E[ \mu_t(\theta) \frac{d}{dt} v_t(\theta) ] \| dt \leq \left(1+ 2 k (\gamma+2) / \gamma^2 \right) r \,.
\end{equation}
We also have (since $y_t(\theta) \in C$):
\begin{equation} \label{eq:int2}
\int_0^T \| \E[ v_t(\theta) \frac{d}{dt} \mu_t(\theta) ] \| dt \leq \gamma' r \E \int_0^T | \frac{d}{dt} \mu_t(\theta) | dt \,.
\end{equation}
Let us calculate:
\begin{equation} \label{eq:intmu}
\frac{d}{dt} \mu_t(\theta) = \left(\frac{d}{dt} (m_t - u_t(\theta)) \right)^{\top} \Sigma_t^{-1} m_t + (m_t - u_t(\theta))^{\top} \left(\frac{d}{dt} \Sigma_t^{-1} \right) m_t + (m_t - u_t(\theta))^{\top} \Sigma_t^{-1} \frac{d}{dt} m_t \,.
\end{equation}
The first and last term in \eqref{eq:intmu} are easy to control using that $m_t$ and $u_t(\theta)$ have total movement $\leq r$, together with \eqref{eq:lbsigma} and \eqref{eq:ubsimple}. Precisely we get that the contribution of those two terms to the integral in the right hand side of \eqref{eq:int2} is at most:
\begin{equation} \label{eq:int3}
\frac{2 k}{\gamma^2 r^2} \times (2 r^2 + (\gamma +2) r^2) = 2 k (\gamma +4) / \gamma^2 \,.
\end{equation}
For the middle term in \eqref{eq:intmu} one has 
\[
\frac{d}{dt} \Sigma_t^{-1} = - \Sigma_t^{-1} \left( \frac{d}{dt} \Sigma_t \right) \Sigma_t^{-1} = 2 \Sigma_t^{-1} \E \left[ (u_t(\theta) - m_t) \left( \frac{d}{dt} (u_t(\theta) - m_t) \right)^{\top} \right] \Sigma_t^{-1} \,.
\]
In particular we have that the spectral norm of $\frac{d}{dt} \Sigma_t^{-1}$ is bounded by 
\[
2 \times \left(\frac{2 k}{\gamma^2 r^2}\right)^2 \times (\gamma +2) r \times \E \left[ \| \frac{d}{dt} (u_t(\theta) - m_t) \| \right] \,  
\]
so that the contribution of the middle term in \eqref{eq:intmu} to the integral in the right hand side of \eqref{eq:int2} is at most:
\begin{equation} \label{eq:int4}
(\gamma +2) r^2 \times 2 \times (2 k / (\gamma^2 r^2))^2 \times (\gamma +2) r  \times r = 2 \left( 2 k (\gamma+2) / \gamma^2 \right)^2 \,.
\end{equation}
Putting together \eqref{eq:int3} and \eqref{eq:int4} with \eqref{eq:int2}, and adding \eqref{eq:int1}, one finally obtains:
\[
\frac{1}{r} \int_0^T \|\frac{d}{dt} Y_t\| dt \leq 1+ 2k (\gamma+2) / \gamma^2 + \gamma' \left( 2k(\gamma +4) / \gamma^2 + 2 \left( 2 k (\gamma+2) / \gamma^2 \right)^2 \right) \,,
\]
which concludes the proof (up to a straightforward numerical verification using $\gamma\geq 12 k$).
\end{proof}

\subsection{Proof of Theorem \ref{thm:main}} \label{sec:final}
Let us apply Lemma \ref{lem:key} with $\gamma' = \alpha$ and $\gamma = \frac{\alpha}{16 k^2 (k+1)}$, and consider ending the phase with $\zeta_k = 65 (k+1)^4 \geq \left(1 + k \frac{2+4 \gamma'}{\gamma}\right)$. Note that to satisfy the condition $\gamma \geq 12 k$ we can take
\begin{equation} \label{eq:alphadef}
\alpha = 192 (d+1)^4 \,,
\end{equation}
which also satisfies \eqref{eq:alphacond}. We obtain that when the lower dimensional phase ends, we have a point $\theta \in V^{\perp}$ such that $\|\theta\| = \gamma r \leq \frac{\delta(P_{V^{\perp}}(\Omega))}{8 k (k+1)} r$ and $\theta \not\in P_{V^{\perp}} {\Omega}'_T$. Thus by Lemma \ref{lem:grun} we obtain that (recall that $\Omega^{\mathrm{new}} = \Omega \cap (\Omega_T' + r B)$)
\[
\vol(P_{V^{\perp}} \Omega^{\mathrm{new}}) \leq \frac{9}{10} \vol(P_{V^{\perp}} \Omega) \,.
\]
We can now apply Lemma \ref{lem:betadef} with
\[
\beta  = \sum_{k=1}^{d} 50 k \log(k+1) \cdot (\zeta_k \cdot \omega_{d-k} + 7 \alpha + 1) \leq (7 d)^7 + \sum_{k=1}^d (6 k)^6 \omega_{d-k} \leq \sum_{k=1}^d (14 k)^6 \omega_{d-k} \,.
\]
and plug it into Lemma \ref{lem:conc} to obtain:
\[
\omega_d \leq \sum_{k=1}^d (24 k)^6 \omega_{d-k}  \,.
\]
The proof of Theorem \ref{thm:main} is now concluded with a straightforward numerical verification. 
\newpage
\begin{algorithm2e}[H]

\caption{$\mathtt{Chasing}(x_0)$}

\SetAlgoLined\DontPrintSemicolon

\textbf{Input:} starting point $x_0 \in \R^d$.

Receive a first request $K \subset \R^d$ and set $r=\mathrm{dist}(x_0, K)$.

Run $\mathtt{GlobalPhase}(x_0,r)$.

\end{algorithm2e}

\begin{algorithm2e}[H]

\caption{$\mathtt{GlobalPhase}(z,r)$}

\SetAlgoLined\DontPrintSemicolon

\textbf{Input:} scale $r>0$ and center $z \in \R^d$.

\textbf{Initialization:} $\alpha \gets 192 (d+1)^4$, $R \gets 7 \alpha r$, $\Omega \gets B(z,R)$, $V \gets \{0\}$, $x \gets z$. \hspace{1in} \tcp*[f]{$x$ represents the online selector's location}

\While{$V \neq \mathbb{R}^d$}
{
Receive request $K \subset \R^d$ such that $x \not\in K$.

\If(\tcp*[f]{Lemma \ref{lem:emptycase}}){$\mathrm{dist}(z,K) > R+r$}{

End call and start $\mathtt{GlobalPhase}(z, 2r)$ with $K$ as first request.}

Move $x$ to the closest point in $K$. \tcp*[f]{Satisfy the request}

Let $\mathrm{cost}^*$ be the minimum cost to service the requests $K \cap B(z,R)$ so far.

\If (\tcp*[f]{[Subsection \ref{sec:padded}, (iii)]}){$\mathrm{cost}^* \geq r$}{End call and start $\mathtt{GlobalPhase}(z, 2r)$.}

Start $\mathtt{LowerDimPhase}(\Omega, V)$ with $K$ as first request, and obtain $\Omega'$.

$\Omega \gets \Omega \cap \{\Omega' + rB\}$

\While{there
is $v \in V^{\perp}$ such that $\delta_v(\Omega) < \frac{\alpha}{2 \dim V^{\perp}}$}{

$V\gets \mathrm{span}(v,V)$.

}

Move $x$ to the closest point such that $P_{V^{\perp}} x = \mathrm{cg}(P_{V^{\perp}} \Omega)$.
}
\uIf(\tcp*[f]{Condition 1 in Subsection \ref{sec:padded}}){$\Omega \subset \mathring{B}(z,R-r)$}{
End call and start $\mathtt{GlobalPhase}(z^{\mathrm{new}}, r/2)$ where $z^{\mathrm{new}}$ is a point s.t. $\Omega \subset B(z^{\mathrm{new}}, \alpha r)$.}
\Else(\tcp*[f]{Condition 2 in Subsection \ref{sec:padded}})
{
End call and start $\mathtt{GlobalPhase}(z, 2r)$.
}
\end{algorithm2e}

\begin{algorithm2e}[H]

\caption{$\mathtt{LowerDimPhase}(\Omega,V, r)$}

\SetAlgoLined\DontPrintSemicolon

\textbf{Input:} current set $\Omega \subset \R^d$, subspace $V \subset \R^d$, scale $r$, and request sequence $K_t$.

\textbf{Initialization:} $\zeta \gets 65 (\mathrm{dim}(V^{\perp})+1)^2$

\uIf{$\mathrm{dim}(V) = 0$}{Receive request $K_1$ and \textbf{Return} $\Omega'=K_1$}
\Else{
Let $c \gets \mathrm{cg}(P_{V^{\perp}} \Omega)$ and run $\mathtt{Chasing}(c)$ on the request sequence $K_t' = K_t\cap (c+V)$ until a total movement of $\zeta \cdot \omega_{\mathrm{dim}(V)} \cdot r$.

\textbf{Return} $\Omega'$ where \tcp*[f]{Used Lemma \ref{lem:key} with the fact $\Omega\subset C$}
\[
\Omega' = \left\{y \in \Omega \text{ s.t. } \exists (y_t)_{t \in \{1,\hdots,T\}} \text{ with } y_{T} = y , y_t \in K'_t \cap \Omega, \text{ and } \sum_{t=1}^{T-1} \|y_t - y_{t+1}\| \leq r \right\}.
\]
}

\end{algorithm2e}
 \label{sec:alg}

\newpage
\bibliographystyle{plainnat}
\bibliography{bib}

\end{document}